\documentclass[reqno,11pt]{amsart}

\usepackage[T1]{fontenc} 
\usepackage[utf8]{inputenc} 
\usepackage{mathtools}
\usepackage{amssymb}
\usepackage{amsthm}
\usepackage{bbm}
\usepackage{mathrsfs}
\usepackage[english]{babel} 
\usepackage{enumitem}
\usepackage{dsfont}
\usepackage{url}
\usepackage{lipsum} 

\usepackage{hyperref}
\hypersetup{linktocpage=true, colorlinks=true, breaklinks=true, linkcolor=blue, menucolor=black, urlcolor=black,citecolor=blue,filecolor=green}

\newcommand{\N}{\mathbb{N}} 
\newcommand{\R}{\mathbb{R}} 

\newcommand{\1}{\mathds{1}}

\theoremstyle{plain}
\newtheorem{thrm}{Theorem}
\newtheorem{prop}[thrm]{Proposition}
\newtheorem{cor}[thrm]{Corollary}

\theoremstyle{definition}

\newtheorem{remark}[thrm]{Remark}

\begin{document}			

\title[Renewal approach for the energy-momentum relation]{Renewal approach for the energy-momentum relation of the Fröhlich polaron}
			
\author{Steffen Polzer}
\address{Steffen Polzer \hfill\newline
	\indent Section de mathématiques, Université de Genève
}
\email{steffen.polzer@unige.ch}
			
\maketitle

\begin{quote}
{\small
	{\bf Abstract.}
	We study the qualitative behaviour of the energy-momentum relation of the Fröhlich polaron at fixed coupling strength. Among other properties, we show that it is non-decreasing and that the correction to the quasi-particle energy is negative. We give a proof that the effective mass lies in $(1, \infty)$ that does not need the validity of a central limit theorem for the path measure. 
} 	
\end{quote}
			
\section{Introduction and results}
The polaron models the interaction of an electron with a polar crystal. The Fröhlich Hamiltonian describing the interaction of the electron with the lattice vibrations has a fiber decomposition in terms of the Hamiltonians
\begin{equation*}
H(P) = \frac12 (P - P_f)^2 + \mathbf N + \frac{\sqrt{\alpha}}{\sqrt{2} \pi} \int_{\mathbb R^3} 
\frac{1}{|k|}(a_k + a^\ast_k) \, \mathrm dk
\end{equation*}
at fixed total momentum $P\in \mathbb R^3$ that act on the bosonic Fock space over $L^2(\R^3)$. Here $a_k^\ast$ and $a_k$ are the creation and annihilation operators satisfying the canonical commutation relations $[a^\ast_k, a_{k'}] = \delta(k-k')$, $\mathbf N \equiv  \int_{\mathbb R^3} a^\ast_k a_k\, \mathrm d k$ is the number operator, $P_f \equiv \int_{\R^3} k a_k^* a_k \mathrm dk$ is the momentum operator of the field and $\alpha>0$ is the coupling constant. Of particular interest has been the energy momentum relation
\begin{equation*}
	E(P) \coloneqq \inf \operatorname{spec}(H(P)).
\end{equation*}
For small $|P|$ the system is believed to behave like a free particle with an increased ``effective mass''. $E$ is known to have a strict local minimum at $P=0$ and to be smooth in a neighbourhood of the origin. The effective mass is defined as the inverse of the curvature at the origin such that
\begin{equation}
	\label{Equation: Asympotics around 0}
	E(P) - E(0) = \frac{1}{2m_{\text{eff}}} |P|^2 + o(|P|^2)
\end{equation}
in the limit $P \to 0$. While significant effort has been put into the study of the asymptotic behaviour of $E(0)$ and $m_{\text{eff}}$ in the strong coupling limit $\alpha \to \infty$ (see e.g \cite{DoVa83}, \cite{LiTh97}, \cite{LiSe20}, \cite{BP22}, \cite{MMS22}), we will be interested in the qualitative behaviour of $E$ at a fixed value of the coupling constant. One valuable tool for the analysis of $E(0)$ and $m_{\text{eff}}$ has been their probabilistic representation obtained via the Feynman-Kac formula. The approach taken below extends the probabilistic methods developed for the analysis of the effective mass to the whole energy momentum relation.\\
\\ 
Let
\begin{equation*}
	E_\text{ess}(P)  \coloneqq \inf  \text{ess spec}(H(P))  
\end{equation*}
be the bottom of the essential spectrum. It is known \cite{Sp88} that
\begin{equation*}
	E_\text{ess}(P) = E(0) + 1
\end{equation*}
for all $P$. From now on, we will often abuse notation and identify a radially symmetric function on $\mathbb R^3$ with a function on $[0, \infty)$. Keeping that in mind, let
\begin{equation*}
	\mathcal I_0 \coloneqq \{P\in [0, \infty): E(P) <   E(0) + 1\}
\end{equation*}
(which is known to contain a neighbourhood of the origin). For Hamiltonians with stronger regularity assumptions (e.g. the Fröhlich polaron with an ultraviolet cutoff) it is known \cite{Mo06} that the spectral gap closes in the limit, i.e. that $\lim_{P \to \infty} E(P) = E_\text{ess}(0)$. For the Fröhlich polaron in dimensions 1 and 2 it is known \cite{Sp88} that $\mathcal I_0 = [0, \infty)$ i.e. that the spectral gap does not close in a finite interval. In dimension 3, however, it has been predicted in the physics literature that $\mathcal I_0$ is bounded \cite{Fe72}. For sufficiently small coupling constants, this has been shown in \cite{Da17}.  In the framework presented below, the question whether $\mathcal I_0$ is bounded or unbounded reduces to the study of the tails of a probability distribution on $(0, \infty)^2$. There does not seem to be known much about the behaviour of $E$ in the intermediate $P$-regime. In \cite{DySp20} it was shown that $E$ is real analytic on $\mathcal I_0$ with $E(0) \leq E(P)$ for all $P$ and that the inequality is strict for $P$ outside of a compact set. In recent work it has been shown \cite{LMM22} that $E$ has indeed a strict global minimum in 0. In the present text, we will prove some previously unknown properties of $E$, namely monotonicity and concavity of $P\mapsto E(\sqrt{P})$ on $[0, \infty)$, both of which are strict on $\mathcal I_0$.  
The (strict) monotonicity additionally allows us to replicate the result of \cite{LMM22}.
\begin{thrm}
	\label{Theorem: Main result}
	The following holds.
	\begin{enumerate}[label=(\roman*)]
		\item $P \mapsto E(P)$ is non-decreasing on $[0, \infty)$ and strictly increasing on $\mathcal I_0$. In particular, $\mathcal I_0$ is an (potentially unbounded) interval.
		\item $P \mapsto E(\sqrt{P})$ is strictly concave on $\mathcal I_0$. In particular
		\begin{equation*}
			E(P) - E(0) < \frac{1}{2m_{\text{eff}}} P^2
		\end{equation*}
		for all $P>0$, i.e. the correction to the quasi-particle energy is negative and  $\big[0, \sqrt{2 m_\text{eff}}\, \big) \subset \mathcal I_0$.
		\item For $|P| \notin \operatorname{cl}(\mathcal I_0)$ we have $\lim_{\lambda \uparrow E(P)}\langle \Omega, (H(P) - \lambda)^{-1} \Omega \rangle<\infty$, where $\Omega$ is the Fock vacuum, in particular $H(P)$ does not have a ground state.
	\end{enumerate}
\end{thrm}
For the polaron with an ultraviolet cut-off and in dimensions 3 and 4, the non-existence of a ground state for $|P|\notin \mathcal I_0$ has been shown in \cite{Mo06}. In a certain limit of strong coupling, the negativity of the correction to the quasi-particle energy has been shown in \cite{MMS22}.
In (iii) we used that if $H(P)$ has a ground state, then it is non-orthogonal $\Omega$: The operator $e^{\mathrm i \pi \mathbf N} e^{-TH(P)} e^{-\mathrm i \pi \mathbf N}$ is for all $T>0$ positivity improving \cite[Theorem 6.3]{Miy10} which in turn implies that if there exists a ground $\psi_P$ state of $H(P)$ then it is unique (up to a phase) and can be chosen such that $e^{\mathrm i \pi \mathbf N}\psi_P$ is strictly positive \cite[Theorem 2.12]{Miy10}. In  \cite[Theorem 6.4]{Miy10} it was shown that there exists a ground state of $H(P)$ for $|P|< \sqrt{2}$. Part (ii) of our Theorem \ref{Theorem: Main result} allows us to improve this to existence of a ground state for $|P|< \sqrt{2m_\text{eff}}$.\\
\\
Before starting with the proof of Theorem \ref{Theorem: Main result}, we give a brief summary of our approach. An application of the Feynman-Kac formula to the semigroup generated by the Hamiltonian yields \cite{DySp20}
\begin{equation*}
	\label{Equation: Feyman-Kac}
	\langle \Omega, e^{-TH(P)} \Omega \rangle = \int_{C([0, \infty), \mathbb R^3)} \mathcal W(\mathrm dX)\, e^{- \mathrm i P \cdot X_T} \exp\bigg( \frac{\alpha}{2} \int_0^T \int_0^T \mathrm ds \mathrm dt \, \frac{e^{-|t-s|}}{|X_{s, t}|} \bigg)
\end{equation*}
for all $P\in \mathbb R^3$ and $T\geq 0$, where $\Omega$ is the Fock vacuum, $\mathcal W$ is the distribution of a three dimensional Brownian motion started in the origin and $X_{s, t} \coloneqq X_t -X_s$ for $X\in C([0, \infty), \R^3)$ and $s, t\geq 0$. After normalizing the expression above by dividing by $\langle \Omega, e^{-TH(0)} \Omega \rangle$, one can study $E$ by looking at the large $T$ asymptotics of Brownian motion perturbed by a pair potential. Herbert Spohn conjectured in \cite{Sp87} convergence of the resulting path measure under diffusive rescaling to Brownian motion and showed that the respective diffusion constant is then the inverse of the effective mass, see also \cite{DySp20}. The validity of this central limit theorem was shown by Mukherjee and Varadhan in \cite{MV19} for sufficiently small $\alpha$ and then, by extending the proof given in \cite{MV19}, for all $\alpha$ in \cite{BP21}. The proof given in \cite{MV19} relies on a representation of the path measure as a mixture of Gaussian measures, where the mixing measure can be expressed in terms of a perturbed birth and death process. An application of renewal theory then yielded the existence of an infinite volume measure and a central limit theorem provided that a certain technical condition holds whose validity was proven for sufficiently small coupling parameters. In \cite{BP21} this approach was continued. Rather than directly verifying the validity of said condition, the point process representation was used in order to derive a renewal equation for $T \mapsto \langle \Omega, e^{-TH(0)} \Omega \rangle$. It was then shown that the condition of Mukherjee and Varadhan is equivalent to the known existence of a ground state of $H(0)$ that is non-orthogonal to $\Omega$. We will use a similar approach and derive renewal equations for $T \mapsto \langle \Omega, e^{-TH(P)} \Omega \rangle$ for any $P$. We arrive at our results by comparing the asymptotic behaviour of the solutions in dependency of $P$. \\
\\
In our units the free electron has mass 1 and physically one would expect that
	$1 < m_{\text{eff}} < \infty$.
The proof of the central limit theorem entails a formula for the diffusion constant that directly implies that this indeed holds. We will give an additional proof that yields (essentially) the same formula for the effective mass
but that does not rely on the validity of a central limit theorem. Numerous efforts have been made to establish central limit theorems for related models (see e.g. \cite{BeSp04}, \cite{Gu06}, \cite{Mu22}, \cite{BP21}) and a generalization of the method presented below may be a viable alternative to study the effective mass with probabilistic methods.
			
\section{Proof of Theorem \ref{Theorem: Main result}}		
We define $\triangle \coloneqq \{(s, t)\in \mathbb [0, \infty)^2:\, s<t\}$ and $\mathcal Y \coloneqq \bigcup_{n=0}^\infty (\triangle \times [0, \infty))^n$, and equip the latter with the disjoint-union $\sigma$-algebra (i.e. the final $\sigma$-algebra with respect to the canonical injections $(\triangle \times [0, \infty))^n \hookrightarrow \mathcal Y$, $n\in \N$). For $\zeta = ((s_i, t_i, u_i))_{1\leq i \leq n}\in \mathcal Y$ let 
\begin{equation*}
	T_1(\zeta) \coloneqq \sup_i t_i, \quad \sigma^2(\zeta) \coloneqq \operatorname{dist}_{L^2}\Big(B_{T_1(\zeta)}, \operatorname{span}\{u_i B_{s_i, t_i} + Z_i: \, 1\leq i \leq n\}\Big)^2
\end{equation*}
where $(B_t)_{t\geq 0}$ is a one dimensional Brownian motion and $(Z_n)_n$ is an iid sequence of $\mathcal N(0, 1)$ distributed random variables that is independent of $(B_t)_{t\geq 0}$. For a measure $\mu$ on $\mathcal Y$ and a measurable function $f:\mathcal Y \to \mathbb R$ we abbreviate  $\mu(f) \coloneqq \int_{\mathcal Y} \mu(\mathrm d\zeta) f(\zeta)$ provided that the integral exists in $\overline{\mathbb R}$. Additionally, we set $f_P(T) \coloneqq \langle \Omega, e^{-TH(P)} \Omega \rangle$ for $P\in \mathbb R^3$, $T\geq 0$.
 
\begin{prop}
	\label{Proposition: Our renewal equations}
	There exists a measure $\mu$ on $\mathcal Y$ such that
	\begin{equation*}
		f_P(T) = \mu\big(e^{-P^2 \sigma^2/2} f_P(T-T_1)\1_{\{T_1 \leq T\}}\big) + e^{-P^2T/2}
	\end{equation*}
	holds for all $P \in \mathbb R^3$ and $T \geq 0$.	
\end{prop}

\begin{proof}
	We define
	\begin{equation*}
		F_P(T_1, T_2, \xi) \coloneqq \int \mathcal W(\mathrm d X) \, e^{- \mathrm i P \cdot X_{T_1, T_2}} \prod_{i=1}^{n} |X_{s_i, t_i}|^{-1}
	\end{equation*}
	for $T_1, T_2 \geq 0$ and $\xi = ((s_i,t_i))_{1\leq i \leq n}$ such that the integral is well defined.
	Let $\nu_T(\mathrm ds \mathrm dt) \coloneqq \alpha e^{-|t-s|} \1_{\{0<s<t<T\}} \mathrm ds \mathrm dt$. Expanding the exponential into a series and exchanging the order of integration leads to\footnote{Using that the integral is finite for $P=0$ shows that this is indeed justified.}
		\begin{align}
		\label{Equation: Paths measure by PPP}
		f_P(T) &=  \int \mathcal W(\mathrm d X) \, e^{-\mathrm i P \cdot X_{0, T}} \exp \left( \int \nu_T(\mathrm ds \mathrm dt) |X_{s, t}|^{-1}  \right) \nonumber \\ \nonumber
		&= \sum_{n=0}^\infty \frac{1}{n!} \int \nu_{T}^{\otimes n}(\mathrm d s_1 \mathrm d t_1, \hdots, \mathrm d s_n \mathrm d t_n) \int \mathcal W(\mathrm dX)\, e^{-\mathrm i P \cdot X_{0, T}} \prod_{i=1}^{n}|X_{s_i, t_i}|^{-1} \nonumber \\
		&= e^{c_{T}} \int \Gamma_{T}(\mathrm d\xi) F_P(0, T, \xi)
	\end{align}
	where $\Gamma_T$ is the distribution of a Poisson point process on $\mathbb R^2$ with intensity measure $\nu_T$ and $c_T \coloneqq \nu_T(\mathbb R^2)$. Let $\Gamma$ be the distribution of a Poisson point process on $\mathbb R^2$ with intensity measure $\nu(\mathrm ds \mathrm dt) \coloneqq \alpha e^{-|t-s|} \1_{\{0<s<t\}} \mathrm ds \mathrm dt$. The measure $\Gamma$ can be seen as the distribution of a birth and death process with birth rate $\alpha$ and death rate 1 (started with no individual alive at time 0) by identifying an individual that is born at $s$ and that dies at $t$ with the point $(s, t)$. For $t \geq 0$ and a configuration $\xi = ((s_i, t_i))_{i}$ of individuals let $N_t(\xi) \coloneqq |\{i: \, s_i \leq t < t_i\}|$ be the number of individuals alive at time $t$. By the restriction theorem for Poisson point processes, $\Gamma_{T}$ can be obtained by restricting $\Gamma$ to the process of all individuals that are born before $T$ conditional on the event that no individual is alive at time $T$. One can easily verify that 
	\begin{equation*}
		e^{c_T} = e^{\alpha T} e^{-\nu([0, T] \times (T, \infty))} = e^{\alpha T} \Gamma(N_T = 0).
	\end{equation*}
	Hence, if we denote by $\xi_{t_1, t_2}$ the restriction of $\xi$ to all individuals born in $[t_1, t_2)$, we can rewrite
	\begin{equation*}
		f_P(T) =  e^{\alpha T} \int \Gamma(\mathrm d\xi) F_P(0, T, \xi_{0, T}) \1_{\{N_T(\xi) = 0\}}.
	\end{equation*}
	Let
	\begin{equation*}
		\tau(\xi) \coloneqq \inf\{t\geq \inf_i s_i: \, N_t(\xi) = 0\}
	\end{equation*}
	be the first time after the first birth at which no individual is alive. By independence of Wiener increments
	\begin{equation*}
		F_P(0, T, \xi_{0, T}) = F_P(0, \tau(\xi), \xi_{0, \tau(\xi)}) F_P(\tau(\xi), T, \xi_{\tau(\xi), T}). 
	\end{equation*}
	for all $\xi \in \{\tau \leq T\}$ such that that the left hand side is well defined.
	Let $\Xi$ be the distribution of $\xi \mapsto \xi_{0, \tau(\xi)}$ under $\Gamma$. The process $\Gamma$ regenerates after $\tau$ and by the translation invariance of $F_P$ under a simultaneous time shift in all variables and since $e^{\alpha T} = e^{\alpha \tau}e^{\alpha (T-\tau)}$
	\begin{align*}
		f_P(T) =& \int \Xi(\mathrm d \xi) \1_{\{\tau(\xi) \leq T\}}e^{\alpha \tau(\xi)} F_P(0, \tau(\xi), \xi) f_P(T-\tau(\xi)) \\ 
		&+  e^{\alpha T}\int \Xi(\mathrm d \xi) \1_{\{\tau(\xi) > T, \, N_T(\xi) = 0\}} F_P(0, T, \xi_{0, T}).
	\end{align*}
	The event $\{\tau > T, \, N_T = 0\}$ happens if and only if there is no birth until time $T$. Then $\xi_{0, T}$ is the empty configuration and hence
	\begin{align*}
		F_P(0, T, \xi_{0, T}) = \mathbb E_\mathcal W\big[e^{-\mathrm i P \cdot X_T}\big] = e^{-P^2 T/2}
	\end{align*}
	for $\xi \in \{\tau> T, \, N_T = 0\}$. Under $\Xi$, the time until the first birth is $\operatorname{Exp}(\alpha)$ distributed and hence $	\Xi(\tau> T, \, N_T = 0) = e^{-\alpha T}$.
	Combined, this gives us
	\begin{equation*}
		f_P(T) = e^{-P^2T/2} + \int \Xi(\mathrm d \xi) e^{\alpha \tau(\xi)} \1_{\{\tau(\xi) \leq T\}} F_P(0, \tau(\xi), \xi) f_P(T-\tau(\xi)).
	\end{equation*}
	For $(\xi, u) \in \triangle^n \times [0, \infty)^n$ we define $\mathbb P_{\xi, u}$ by
	\begin{equation*}
		\mathbb P_{\xi, u}(\mathrm dX) \coloneqq \frac{1}{\phi(\xi, u)} e^{-\sum_{i=1}^n u_i^2 |X_{s_i, t_i}|^2/2} \mathcal W(\mathrm dX)
	\end{equation*}
	where $\phi(\xi, u)$ is a normalization constant.
	Then $\mathbb P_{\xi, u}$ is a centred and rotationally symmetric Gaussian measure and
	\begin{equation*}
		\frac{1}{3}\mathbb E_{\mathbb P_{\xi, u}}\big[|X_t|^2\big] = \operatorname{dist}_{L^2}\Big(B_{t}, \operatorname{span}\{u_i B_{s_i, t_i} + Z_i: \, 1\leq i \leq n\}\Big)^2 \eqqcolon \sigma^2_t(\xi, u)
	\end{equation*}
	for all $t\geq 0$, see the proof of Proposition 3.2 in \cite{BP22}. We thus have
	\begin{align*}
		F_P(0, t, \xi) &= \int \mathcal W(\mathrm dX) \int_{[0, \infty)^n}  \mathrm du\, (2/\pi)^{n/2}\, e^{-\mathrm i P \cdot X_t} e^{-\sum_{i=1}^n u_i^2|X_{s_i, t_i}|^2/2} \\
		&= \int_{[0, \infty)^n} \mathrm du\, (2/\pi)^{n/2} \phi(\xi, u) e^{-P^2 \sigma_t^2(\xi, u)/2}.
	\end{align*}
	Hence, the measure we are looking for is given by
	\begin{equation}
		\mu(\mathrm d \xi \mathrm du) \coloneqq \Xi(\mathrm d \xi) \mathrm du \, (2/\pi)^{n(\xi)/2} e^{\alpha \tau(\xi)} \phi(\xi, u)
	\end{equation}
	under the identification of $\triangle^n \times [0, \infty)^n$ with $(\triangle \times [0, \infty))^n$.
\end{proof}

\begin{prop}
	\label{Propposition: Matrix element of Resolvent}
	We have $\mu(e^{-\sigma^2 P^2 /2 + E(P) T_1}) \leq 1$ for all $P\in  \R^3$ and for $\lambda < E(P)$ we have
	\begin{equation*}
		\langle \Omega, (H(P) - \lambda)^{-1} \Omega \rangle = \frac{1}{P^2/2-\lambda} \cdot \frac{1}{1-\mu(e^{-P^2\sigma^2 /2 + \lambda T_1})}.
	\end{equation*}
	If $|P|\in \mathcal I_0$ then $E(P)$ is the unique real number satisfying
	\begin{equation*}
		\mu(e^{-P^2\sigma^2 /2 + E(P) T_1}) = 1.
	\end{equation*}
\end{prop}

\begin{proof}
	For $P\in \R^3$, let $\nu_P$ be the image measure of $e^{-P^2\sigma^2(\zeta)/2}\mu(\mathrm d \zeta)$ under the map $T_1$ and let $z_P(T) \coloneqq e^{-P^2T/2}$ for $T \geq 0$. By Proposition \ref{Proposition: Our renewal equations}, for any $P \geq 0$ the renewal equation
	\begin{equation}
		\label{Equation: Renewal equation}
		f_P = \nu_P*f_P + z_P
	\end{equation}
	holds, where the convolution $\nu_P*f_P$ is defined as
	\begin{equation*}
		(\nu_P*f_P)(T) \coloneqq \int_{[0, T]} \nu_P(\mathrm dt) f_P(T-t)
	\end{equation*}
	for $T\geq 0$. As $f_P$ is continuous and strictly positive, $\inf_{0 \leq t \leq T} f_P(t) >0$ and hence the measure $\nu_P$ is locally finite. Renewal theory implies that the unique locally bounded solution to \eqref{Equation: Renewal equation} is given by
	\begin{equation*}
		f_P = \sum_{n=0}^\infty \nu_P^{*n}*z_P
	\end{equation*}
 	Taking the Laplace transform leads to
	\begin{equation*}
			\langle \Omega, (H(P) - \lambda)^{-1} \Omega \rangle = \mathcal L(f_P)(-\lambda) =  \frac{1}{P^2/2-\lambda} \sum_{n=0}^\infty \mathcal L(\nu_P)^n(-\lambda)
	\end{equation*}
	for\footnote{The inequality $E(P)<P^2/2$ follows from the considerations above and can also be obtained directly from the definition of the Hamiltonians by using $E(0)<0$ and the estimate $E(P) \leq \langle \psi_0, H(P) \psi_0 \rangle$, where $\psi_0$ is the ground state of $H(0)$.} $\lambda < E(P)$. In particular, $\mathcal L(\nu_P)(-\lambda)<1$ for $\lambda < E(P)$ and
		\begin{equation*}
		\langle \Omega, (H(P) - \lambda)^{-1} \Omega \rangle =  \frac{1}{P^2/2-\lambda} \cdot \frac{1}{1-\mu(e^{-P^2 \sigma^2/2 + \lambda T_1})}.
	\end{equation*}
	As mentioned earlier, if there exists a ground state of $H(P)$ then it is unique and non-orthogonal to $\Omega$. In combination with the spectral theorem this implies for $|P|\in \mathcal I_0$ that
	\begin{equation*}
		\lim_{\lambda \uparrow E(P)} \langle \Omega, (H(P) - \lambda)^{-1} \Omega \rangle = \infty
	\end{equation*}
	and hence $\mu(e^{-\sigma^2 P^2 /2 + E(P) T_1}) = 1$ by the monotone convergence theorem.
\end{proof}
\begin{remark}
	Let $P\in \mathbb R^3$ such that $|P|\in \mathcal I_0$ and $\psi_P$ be the unique ground state of $H(P)$. By an application of the spectral theorem 
	\begin{equation*}
		\lim_{T \to \infty} f_P(T)e^{TE(P)} = \lim_{T \to \infty}  \langle \Omega, e^{-T(H(P)-E(P))} \Omega \rangle = |\langle \Omega, \psi_P \rangle|^2.
	\end{equation*}
	On the other hand, the limit $\lim_{T \to \infty} f_P(T)e^{TE(P)}$ can be calculated by using the renewal theorem. This gives us the identity
	\begin{equation}
		\label{Equation: Formular for overlap}
		|\langle \Omega, \psi_P \rangle|^2 = \frac{1}{P^2/2 - E(P)} \frac{1}{\mu(T_1 e^{-P^2\sigma^2/2 + E(P)T_1})}.
	\end{equation}
\end{remark}

\begin{cor}
	$E$ is non-decreasing and strictly increasing on $\mathcal I_0$. For $|P| \notin \operatorname{cl}(\mathcal I_0)$ we have $\lim_{\lambda \uparrow E(P)}\langle \Omega, (H(P) - \lambda)^{-1} \Omega \rangle<\infty$  and $H(P)$ does not have a ground state.
\end{cor}

\begin{proof}
	The strict monotonicity on $\mathcal I_0$ follows directly from Proposition \ref{Propposition: Matrix element of Resolvent}. For $P_1, P_2 \notin \mathcal I_0$ with $P_1 < P_2$ we always have
	\begin{equation*}
		\mu(e^{-P_2^2 \sigma^2/2 + E_{\operatorname{ess}}(0)T_1}) < \mu( e^{-P_1^2\sigma^2 /2 + E_{\operatorname{ess}}(0)T_1}) \leq 1
	\end{equation*}
	and hence
	\begin{equation*}
		\mu(e^{-P^2\sigma^2/2 + E_{\operatorname{ess}}(0)T_1})  < 1
	\end{equation*}
	for all $P\in \R^3$ such that $|P| \notin \operatorname{clos}(\mathcal I_0)$. Hence, for those $P$
	\begin{equation*}
		\lim_{\lambda \uparrow E(P)} \langle \Omega, (H(P) - \lambda)^{-1} \Omega \rangle = \frac{1}{P^2/2-E_{\operatorname{ess}}(0)} \cdot \frac{1}{1-\mu(e^{-P^2 \sigma^2 /2 + E_{\operatorname{ess}}(0) T_1})}
	\end{equation*}
	 and $H(P)$ does not have a ground state (since it would need to be non-orthogonal to $\Omega$).
	 If $E$ would be not non-decreasing, then $\mathcal I_0$ would not be an interval i.e. there would exist $P_1 \in [0, \infty) \setminus \mathcal I_0$ and $P_2 \in \mathcal I_0$ such that $P_1 < P_2$. This, however, would imply
	\begin{equation*}
		1 = \mu(e^{-P_2^2\sigma^2/2 + E(P)T_1}) < \mu(e^{-P_1^2\sigma^2/2 + E_\text{ess}(0)T_1}) \leq 1. \qedhere
	\end{equation*}
\end{proof}

\begin{cor}
	\label{Corollary: Is I0 bounded?}
	The interval $\mathcal I_0$ is bounded if and only if there exists a $P\geq 0$ such that
	\begin{equation*}
		\mu(e^{-P^2 \sigma^2/2 + E_\text{ess}(0)T_1}) = \widehat{\mu}(e^{-P^2 \sigma^2/2 + T_1}) < \infty
	\end{equation*}
	where $\widehat{\mu}$ is the probability measure defined by $\widehat{\mu}(\mathrm d \zeta) \coloneqq e^{E(0) T_1(\zeta)} \mu(\mathrm d \zeta)$.	
\end{cor}

\begin{proof}
	This easily follows from the monotone convergence theorem. 
\end{proof}

\begin{cor}
	\label{Corollary: Behavior arround origin}
	We have
	\begin{equation}
		\label{Equation: FOrmular for effective mass}
		m_{\text{eff}} = \frac{\widehat \mu(T_1)}{\widehat \mu(\sigma^2)} \in (1, \infty).
	\end{equation}
\end{cor}

\begin{proof}
	Let $P\in \mathcal I_0$ and $\lambda< E(P)$. Then
	\begin{equation*}
		\mu(e^{-P^2\sigma^2/2 + \lambda T_1}) = 1 - \frac{1}{P^2/2 - \lambda}\cdot \frac{1}{\langle \Omega, (H(P) - \lambda)^{-1} \Omega \rangle}.
	\end{equation*}
	The function $\lambda \mapsto \langle \Omega, (H(P) - \lambda)^{-1} \Omega \rangle^{-1}$ has a removable singularity in $E(P)$ since $E(P)$ is for $P \in \mathcal I_0$ an isolated eigenvalue. This implies that there exists an $\tilde \varepsilon>0$ such that $\mu(e^{-P^2\sigma^2/2 + (E(P) + \tilde \varepsilon) T_1})<\infty$. Since $\sigma^2 \leq T_1$ there thus exist  $\varepsilon, \delta>0$ such that
	\begin{equation*}
	\label{Equation: Differentiation under the integral}
		\mu(e^{-(P-\delta)^2\sigma^2/2 + (E(P) + \varepsilon)T_1}) < \infty.
	\end{equation*}
	 Hence, we may differentiate under the integral. Differentiating
	\begin{equation*}
		1 = \mu(e^{- P^2\sigma^2/2 + E(P) T_1})
	\end{equation*}
	twice with respect to $P$ and evaluating at $P=0$ yields the equality in \eqref{Equation: FOrmular for effective mass}. Notice that both integrals are finite by the previous considerations (or by \eqref{Equation: Formular for overlap} for that matter). Since $\sigma^2 \leq T_1$ and $\mu(\sigma^2 < T_1)>0$ the quotient is strictly larger than 1.	 
\end{proof}

\begin{cor}
	 $P \mapsto E(\sqrt{P})$ is strictly concave on $\mathcal I_0$. In particular
	\begin{equation*}
		E(P) - E(0) < \frac{1}{2m_{\text{eff}}} P^2
	\end{equation*}
	for all $P>0$, i.e. the correction to the quasi-particle energy is negative and  $\big[0, \sqrt{2 m_\text{eff}}\, \big) \subset \mathcal I_0$.
\end{cor}

\begin{proof}
	For $\lambda \in \tilde{\mathcal I_0} \coloneqq \{P^2:\, P \in \mathcal I_0\}$ let $h(\lambda)$ be the unique solution to
	\begin{equation*}
		\mu(e^{-\lambda \sigma^2/2 + h(\lambda) T_1}) = 1,
	\end{equation*}
	i.e. $h = E \circ \sqrt{\cdot}$.
	Then, for $\lambda_1, \lambda_2 \in \tilde{\mathcal I_0}$ with $\lambda_1 \neq \lambda_2$ and $\beta \in (0, 1)$ we get with Hölders inequality with dual exponents $1/\beta$ and $1/(1-\beta)$
	\begin{align*}
	 &\mu(e^{-(\beta \lambda_1 + (1-\beta)\lambda_2) \sigma^2/2 + (\beta h(\lambda_1) + (1-\beta)h(\lambda_2)) T_1}) \\
	 &< \mu(e^{-\lambda_1 \sigma^2/2 + h(\lambda_1) T_1})^{\beta} \mu(e^{- \lambda_2 \sigma^2/2 + h(\lambda_2) T_1})^{1-\beta} = 1
	\end{align*} 
	which means
	$h(\beta \lambda_1  + (1-\beta) \lambda_2) > \beta h(\lambda_1) + (1-\beta)h(\lambda_2)$. Hence, $h$ is strictly concave on $\mathcal I_0$, which implies for all $P \in \mathcal I_0 \setminus \{0\}$
	\begin{equation*}
		E(P) - E(0) = h(P^2) - h(0)  < h'(0) P^2 = \frac{1}{2}E''(0)P^2. \qedhere
	\end{equation*}
\end{proof}

{\bf Acknowledgment:} The author would like to thank David Mitrouskas and Krzysztof Myśliwy for making him aware of some open problems concerning the energy-momentum relation. Additionally, he would like to thank Antti Knowles and Volker Betz for helpful comments on an earlier version of the paper. The author was partially supported by the Swiss
National Science Foundation grant 200020-200400.

\end{document}